\documentclass{article}
\usepackage{fullpage}
\usepackage{amsmath,amsthm,amssymb}

\usepackage{tikz}
\usetikzlibrary{shapes.symbols,arrows,snakes}
\usetikzlibrary{matrix,calc,intersections}

\usepackage{enumerate}
\usepackage{fullpage}
\usepackage{algorithm,algorithmic}

\newcommand{\INITIALLY}{\REQUIRE{}}
\newcommand{\ROUND}{\ENSURE{}}

\newcommand{\nop}[1]{}

\newtheoremstyle{newthm}
  {\topsep}   
  {\topsep}   
  {\itshape}  
  {0pt}       
  {\scshape} 
  {.}         
  {5pt}  
  {}          

\theoremstyle{newthm}

\newtheorem{thm}{Theorem}
\newtheorem{prop}[thm]{Proposition}
\newtheorem{lem}[thm]{Lemma}

\DeclareMathOperator{\hull}{hull}
\DeclareMathOperator{\centroid}{centroid}
\DeclareMathOperator{\vol}{vol}

\DeclareMathOperator{\Cube}{Cb}

\newcommand{\IR}{\mathbb{R}}
\newcommand{\IN}{\mathbb{N}}

\renewcommand{\le}{\leqslant}
\renewcommand{\leq}{\leqslant}
\renewcommand{\ge}{\geqslant}
\renewcommand{\geq}{\geqslant}

\newcommand\blfootnote[1]{%
  \begingroup
  \renewcommand\thefootnote{}\footnote{#1}%
  \addtocounter{footnote}{-1}%
  \endgroup
}

\usepackage{graphicx}

\DeclareMathOperator\In{In}
\DeclareMathOperator\Out{Out}

\title{Multidimensional Asymptotic Consensus in Dynamic Networks}
\author{Bernadette Charron-Bost\textsuperscript{1} \and Matthias
F\"ugger\textsuperscript{2} \and Thomas Nowak\textsuperscript{3}}
\date{
  \textsuperscript{1} \'Ecole polytechnique\\
  \texttt{charron@lix.polytechnique.fr}\\
  \textsuperscript{2} LSV, CNRS \& ENS Paris-Saclay\blfootnote{
  When most of the work was done, Matthias F{\"u}gger was with Max Planck Institute for Informatics.}\\
  \texttt{mfuegger@lsv.fr}\\
  \textsuperscript{3} Universit\'e Paris-Sud\\
  \texttt{thomas.nowak@lri.fr}
}

\begin{document}
\maketitle
\thispagestyle{empty}

\setcounter{footnote}{3}

\begin{abstract}
We study the problem of asymptotic consensus as it occurs in a wide range of applications
  in both man-made and natural systems.
In particular, we study systems with directed communication graphs that may change over time.

We recently proposed a new family of convex combination algorithms in dimension one whose weights depend on the
  received values and not only on the communication topology.
Here, we extend this approach to arbitrarily high dimensions by introducing two new algorithms:
  the ExtremePoint and the Centroid algorithm.
Contrary to classical convex combination algorithms, both
  have component-wise contraction rates that are constant in the number of agents.
Paired with a speed-up technique for convex combination algorithms, we get a convergence time linear
  in the number of agents, which is optimal.

Besides their respective contraction rates, the two algorithms differ in the fact that 
  the Centroid algorithm's update rule is independent of any coordinate system 
  while the ExtremePoint algorithm implicitly assumes a common agreed-upon coordinate system
  among agents.
The latter assumption may be realistic in some man-made multi-agent systems but
  is highly questionable in systems designed for the modelization of natural phenomena.

Finally we prove that our new algorithms also
  achieve asymptotic consensus under very weak connectivity assumptions,
  provided that agent interactions are bidirectional.

\end{abstract}

\newpage
\setcounter{page}{1}

\section{Introduction}

The problem of agents converging to a common final position, also known as
asymptotic consensus, is of utmost importance in a wide range of networking
problems.
One can cite not only artificial, man-made, systems like sensor
fusion~\cite{BS92}, clock synchronization~\cite{LR06}, formation
control~\cite{EH01}, rendezvous in space~\cite{LMA05}, or load
balancing~\cite{Cyb89}, but also the
modelization of natural phenomena like flocking~\cite{VCBCS95}, firefly
synchronization~\cite{MS90}, or opinion dynamics~\cite{HK02}.

Algorithms for asymptotic consensus repeatedly form convex
combinations of their neighbors' positions and move their current position
there.
Classically, these algorithms use weights in the convex combination that only
depend on the communication topology, i.e., the set of the agent's neighbors, but
not on their current positions.
This results in weights that are inversely proportional to the number of
neighbors, e.g., the EqualNeighbor algorithm, the MaxDegree
algorithm~\cite{Mon09}, or the Metropolis algorithm~\cite{XBL05}.
Their analysis consists in studying the stochastic matrices made
up of the weights used by agents in their convex combinations.
An important property of these associated stochastic
matrices is that they inherit irreducibility properties from connectivity properties of the
communication graph.

In the present article, we study the problem of asymptotic consensus in dynamic networks,
in the challenging context of directed communication graphs that may change over time.

In a recent article~\cite{CBFN16}, we proposed a new family of convex
combination algorithms in dimension one whose weights depend on the received values.
A particular example of such an algorithm is the MidPoint algorithm, whose
contraction rate of~$1/2$ is optimal.
The analysis of these algorithms required the development of a new approach
since the graphs associated to the stochastic matrices do not coincide
anymore with the communication graphs and thus do not benefit from the same
connectivity properties.

The goal of the present article is to extend this approach to multiple dimensions.
For this, we present two generalizations of the MidPoint algorithm to the case of an arbitrary dimension~$d$:
the ExtremePoint and the Centroid algorithm.
Contrary to classical convex combination algorithms like EqualNeighbor, both
have component-wise contraction rates that are constant in the number of agents, namely $1-\frac{1}{2d}$ 
for the ExtremePoint algorithm and $1-\frac{1}{d+1}$ for the Centroid algorithm.
Paired with a speed-up technique, we get a convergence time linear
in the number of agents for both algorithms, which is optimal.
	
Besides their respective contraction rates, the two algorithms differ in the fact that 
	the Centroid algorithm's update rule is independent of any coordinate system 
	while the ExtremePoint algorithm implicitly assumes a common agreed-upon coordinate system
	among agents.
The latter assumption may be realistic in some man-made multi-agent systems but
	is highly questionable in systems designed for the modelization of natural phenomena.

The analysis of the two multi-dimensional algorithms that we propose is based on 
	the notion of $\alpha$-safeness.
This property guarantees that every agent stays in the convex hull of its
neighbors and keeps a certain safety margin to its boundary which depends on the
parameter~$\alpha$.
While the proof of safeness for the ExtremePoint algorithm is relatively
straightforward, the proof for the Centroid algorithm uses a Steiner-type
symmetrization and relies on the Brunn-Minkowski inequality.
Apart from its application to asymptotic consensus, this last result may be of
independent geometrical interest.

Finally we prove that our new algorithms share a remarkable property with
the classical asymptotic consensus algorithms. 
Namely convergence is achieved under very weak connectivity
assumptions, provided that agent interactions are bidirectional.
This last point adds to a list of properties of the Centroid algorithm that
makes it a well-suited candidate for the modelization of natural systems~\cite{cha11}.

The paper is organized as follows.
In Section~\ref{sec:model}, we introduce the model and problem statement.
In Section~\ref{sec:one}, we recall results on one-dimensional asymptotic consensus.
Section~\ref{sec:multi_component} generalizes the optimal one-dimensional
algorithm to multiple dimensions in a component-wise fashion.
The coordinate-free Centroid algorithm is presented in Section~\ref{sec:centroid}.
Section~\ref{sec:moreau} presents
  extensions to weaker connectivity assumptions with bidirectional communication
  graphs.
We conclude the paper in Section~\ref{sec:conclusion}.

\section{The Model}\label{sec:model}

We consider a set $[n]= \{1,\dots,n \}$ of agents.
We assume a distributed, round-based computational model in the spirit
     of the Heard-Of model~\cite{CS09}.
Computation proceeds in {\em rounds}: in a round, each agent sends
	its state to its outgoing neighbors, receives messages from its
        incoming neighbors, and finally updates its state according to
	a deterministic local algorithm, i.e., a transition function that maps the collection
	of incoming messages  to a  new state.
Rounds are communication closed in the sense that no agent receives
	messages in round~$t$ that are sent in a round different from~$t$.

Communications that occur in a round are modeled by a directed graph with a node 
	for each agent.
Since an agent can obviously communicate with itself instantaneously, every 
	communication graph contains  a self-loop at each node.	
	
We fix a non-empty set of such directed graphs ${\cal N}$ that determines the 
	{\em network model}.
To fully model dynamic networks in which topology may change continually and 
	unpredictably,  the communication graph at each round  is chosen arbitrarily 
	among ${\cal N}$.
Thus we form the infinite sequences of graphs in~${\cal N}$ which we call {\em communication patterns in\/}
	${\cal N}$.
In each communication pattern, the communication graph at round~$t$ is denoted by~$G_t= \big([n], E_t\big)$, and
$\In_p(t)$ and $\Out_p (t)$ are the sets of incoming and outgoing neighbors (in-neighbors and out-neighbors for short) of
	agent~$p$ in $G_t$.
	
In the following, we use the {\em product\/} of two communication graphs $G$ and $H$, 
	denoted $G \circ H$, which  is the directed graph  with an edge from $p$ to $q$ 
	 if there exists~$r $ such that $(p,r)\in E(G)$ and $ (r,q) \in E(H)$.

\subsection{Asymptotic Consensus}

The state or {\em position\/} of agent~$p$ is captured by a variable $x_p$ in an Euclidean 
	$d$-space, and we let $x_p(t)\in \IR^d$ denote the position of~$p$ at round~$t$.
Thus the $n$-tuple $x(t) = \big( x_1(t), \dots, x_n(t) \big) $ corresponds to
 	the global configuration of the multi-agent system at round~$t$.
We denote the $k^{\text{th}}$ component of $x_p(t)$ by $x_{p,k}(t)$.

We say an algorithm {\em solves asymptotic consensus\/} in a network model ${\cal N}$ if the following holds
	for every initial configuration~$x(0)$ and every communication pattern in ${\cal N}$: 
	\begin{description}
	\vspace{-0.1cm}
	\item{\em Convergence.\/} Each sequence $x_p(t)$ converges.
	\vspace{-0.1cm}
	\item{\em Agreement.\/} If $x_p(t)$ and $x_q(t)$ converge, then they have a common limit.
	\vspace{-0.1cm}
	\item{\em Validity.\/} If $x_p(t)$ converges, then its limit is in the convex hull of the 
	initial states.
	\vspace{-0.1cm}
	\end{description}
	
Our results can be easily translated to the {\em approximate consensus\/} problem, in which
	convergence is replaced by a decision in a finite number of rounds and where
	agreement should be achieved with an arbitrarily small error tolerance (see, e.g., \cite{Lyn96,MH13}).
	
\subsection{Convex Combination Algorithms}	

Because of the validity condition, the natural class of algorithms for solving
	asymptotic consensus  is the class of the {\em convex combination algorithms},
	also called {\em averaging algorithms\/} in the case of dimension one:
	at each round~$t$, every agent~$p$ updates $x_p$ to some convex combination 
	of  the positions it has just received, i.e., the positions of its 
	in-neighbors in the communication graph at round~$t$.
That is 
	\begin{align}
 	x_p(t) = \sum_{q \in \In_p(t)} w_{ p q}(t) \, x_q(t-1) ,\label{eq:update}
	\end{align}
	where weights $w_{p q }(t)$ are  non-negative real numbers with  $\sum_{q \in \In_p(t)} w_{p q}(t) = 1$.
In other words,  at each round~$t$, every agent adopts a new position within the convex
	hull of its in-neighbors in the communication graph~$G_t$.

Since we strive for distributed implementations of convex combination algorithms, $w_{p q }(t)$ is 
        required to be locally computable by $p$.
For example, weights may depend only on the set of $p$'s in-neighbors, as is the case in
        the {\em EqualNeighbor algorithm}, with
	\begin{equation}\label{eq:EN}
	w_{p q}(t) = 1/|\In_p(t)| \, ,
	\end{equation}
	for every in-neighbor~$q$ of $p$.
Weight~$w_{p q}(t)$ may also depend on the positions of the in-neighbors of~$p$, as is the
	case, for instance, with the  update rule
	$$ w_{p q}(t) = \delta _{q  q_0 } \, ,$$	
	where $\delta$ is the Kronecker delta and $q_0$ is one in-neighbor of $p$ in $G_t$
	with the largest first component, i.e., 
	$$x_{q_0,1}(t) = \max \{ x_{q,1}(t) : q \in \In_p(t) \}\, .$$

When  the structure of states allows each agent to record and to relay information it has received during
	any period of $L$ rounds for some positive integer~$L$,   we may be led to  modify time-scale and to
	consider blocks of $L$  consecutive rounds, called {\em macro-rounds}:
	macro-round $s$ is the sequence of rounds $(s-1) L +1, \ldots, s L$ and the corresponding
	information  flow graph, called {\em communication graph at macro-round $s$},  is 
	the product of the communication graphs $ G_{(s -1)L +1}  \circ \ldots \circ G_{ s L } $.
	
\subsection{Solvability of Asymptotic Consensus}
	 
In a previous paper~\cite{CBFN15}, we proved the following characterization of network models in which 
	asymptotic consensus is solvable.

\begin{thm}[\cite{CBFN15}]\label{thm:CBFN15}
In any dimension $d$, the asymptotic consensus problem is solvable in a network model ${\cal N}$ 
	if and only if each graph in ${\cal N}$  has a rooted spanning tree.
\end{thm}

The proof of the sufficient condition of rooted network model is  based on 
	a reduction to {\em nonsplit\/} network models: a directed graph is {\em nonsplit\/}  
	if any two nodes have a common in-neighbor. 
Indeed we showed the following general proposition.

\begin{prop}\label{prop:productrooted}
Every product of $n-1$ rooted graphs  with $n$ nodes and self-loops at all nodes is nonsplit.
\end{prop}

\subsection{Convergence Rate and Convergence Time}

Following~\cite{OT11}, in the case convergence is achieved for some initial configuration~$x(0)$ 
	and some communication pattern, we introduce
	\begin{equation}\label{eq:conv:rate:def}
	\max_{p \in [n]}\ \limsup_{t\to\infty}\  \left\lVert x_p(t) - x_p^* \right\rVert^{1/t}
	\end{equation}
	where $x_p^* = \lim_{t\to\infty} x_p(t) $ and $\lVert . \lVert$ is any norm on $\IR^d$.
This quantity lies in $[0,1]$.
Moreover, it is independent of the norm $\lVert . \lVert$ because of the equivalence of norms in $\IR^d$.
	
For an algorithm that solves asymptotic consensus in a network model~${\cal  N}$, 
	we define its {\em convergence rate\/} $\varrho$ as the supremum
	of~\eqref{eq:conv:rate:def} over all  initial  configurations  and all
	communication patterns with graphs in~${\cal N}$.
	
Regarding approximate consensus and considering the infinity norm on $\IR^n$, we define the 
	{\em convergence time}, $T(\varepsilon)$, by
	\begin{equation}\label{eq:conv:time:def}
	\max_{k \in [d]}\ \inf \left\{ \tau \, : \,  \forall t \geq \tau, \ \delta \big( x^k(t) \big) \leq \varepsilon \, \delta \big( x^k(0) \big)  \right \} 
	\end{equation}
	where $\delta$ is the semi-norm on $\IR^n$ defined by $\delta (u_1, \dots, u_n) = \max_{p\in [n]} (u_p) -\min_{p\in [n]} (u_p)$.
	
\section{The Case of Dimension One}\label{sec:one}

We now briefly present our analysis techniques for the one-dimensional case, which we generalize
        to arbitrary dimensions in Sections~\ref{sec:multi_component} and~\ref{sec:centroid}.
In~\cite{CBFN16}, we proposed a new analysis of the 
	convex combination algorithms in the specific case of dimension one:
	We considered the property of {\em $\alpha$-safeness\/} for averaging algorithms
	which is a generalization of the lower bound condition on positive weights. 
This property focuses on the {\em interval of\/}  transmitted  {\em values\/} and not 
	on the linear  {\em functions\/} (stochastic matrices) applied in the averaging steps, 
	as done classically.
It thus captures the essential properties needed for contracting the range of 
	current values in the system.
This approach led us to propose the first algorithm for asymptotic consensus
	in dynamic rooted networks, with a convergence time that is linear in the 
	number of agents.

\subsection{Nonsplit Network Models}	

Let $\alpha \in ]0,1/2]$;
  	an averaging algorithm is {\em $\alpha $-safe\/} if at any round~$t$,  
	each agent adopts a new value within the interval formed
	by its neighbors in~$G_t$  not too close to the boundary: 
	\begin{equation}\label{eq:safeness}
	\alpha\, M_p(t) + (1-\alpha)\, m_p(t) \leq x_p(t+1) \leq  (1-\alpha) M_p(t) + \alpha \, m_p(t) \, ,
	\end{equation}
	where $m_p(t) = \min_ { q\in \In_p(t+1) } \big( x_q(t) \big) $  and $M_p(t) = \max_ { q\in \In_p(t+1) } \big( x_q(t) \big) $.

Besides, contracting the range of current values in the system is clearly a good mechanism
	to achieve agreement:
	an averaging algorithm is $c$-{\em contracting in\/} ${\cal N}$ if at each round~$t$ 
	of each of its executions with communication patterns in  ${\cal N}$,  
	we have
	$$\delta \big( x(t) \big) \leq c \, \delta \big( x(t-1) \big)  \, .$$
A result from~\cite{Cha13} states that the property of $c$-contraction is also sufficient
	to enforce the convergence of averaging algorithms.
Then the main point  lies in the fact that in a nonsplit network model, 
	an $\alpha$-safe averaging algorithm is $(1-\alpha)$-contracting.
We thus prove the following result.

\begin{thm}\label{thm:nonsplit}
In a nonsplit network model, an $\alpha$-safe averaging algorithm solves asymptotic consensus 
	with a convergence rate $\varrho \leq 1-\alpha$ and a convergence time  
	$T(\varepsilon) \leq \left \lceil \log_{\frac{1}{1-\alpha}} \frac{\delta(0)}{\varepsilon} \right \rceil$.
\end{thm}	

As an immediate consequence of Theorem~\ref{thm:nonsplit}, we obtain that the
	 EqualNeighbor algorithm, that is $(1/n)$-safe,  has a convergence rate bounded
	 by $1-1/n$ and a convergence time in $O\left( n \log  \frac{\delta(0)}{\varepsilon}  \right )$ 
	 in any nonsplit network model.
	 
To improve these bounds, we introduced  the  {\em MidPoint algorithm\/}
	in which weights depend on the set of transmitted values and not on the sole communication graph:
	each agent adopts the mid-point of the range of values it has received, that is
	$$ x_p(t+1) = \frac{m_p(t) + M_p(t)}{2} \, .$$ 
Clearly the MidPoint algorithm is $1/2 $-safe, and so has a maximal contraction rate of $1/2$
	in any nonsplit network model, leading to a convergence rate of $1/2$ and 
	a convergence time  $T(\varepsilon) \leq \left \lceil \log_{2}  \frac{\delta(0)}{\varepsilon} \right \rceil$.

\subsection{Rooted Network Models}

One can easily 	show that if an averaging algorithm is $\alpha$-safe, then it is 
	$\alpha^{L}$-safe with the coarser-grained granularity of macro-rounds
	composed of $L$ consecutive rounds.
Combined with Proposition~\ref{prop:productrooted}, it follows that the EqualNeighbor 
	algorithm solves asymptotic consensus in any rooted network model.
Unfortunately,  the convergence rate and convergence time satisfy
	 $$ 1-\varrho = \Omega \big( n^{-n} \big) \ \mbox{ and } \ T(\varepsilon) = O\left( n^n \log  \frac{\delta(0)}{\varepsilon}  \right )      \, , $$
	 and these exponential  bounds have been proved to be tight~\cite{CBFN15}.

To overcome this time-complexity lower bound of averaging algorithms, we introduced
	the {\em amortization technique\/} \cite{CBFN16} which consists in inserting 
	a value-gathering phase of $n-1$ rounds before each averaging step.
This additional phase transforms $\alpha$-safe algorithms into ``turbo versions'' of themselves
	in that convergence times pass from being exponential to being polynomial in 
	the number of agents.
Amortization assumes implicitly that all agents know the size~$n$ of the network. 
Moreover, it requires  {\em a priori\/} to increase bandwidth channels and local storage capacities by a factor~$n$.
	
In anonymous networks, the amortization technique applies only to averaging algorithms 
	with weights that depend only on the {\em sets\/} of received values without any multiplicity
	concern. 
In contrast to the EqualNeighbor algorithm, MidPoint thus admits an amortized version, 
	called the {\em Amortized MidPoint\/} algorithm.
For its correctness and time-analysis, we just need to observe that the  Amortized MidPoint algorithm 
	reduces to the  MidPoint  algorithm
	 with the granularity of macro-rounds consisting in  blocks of $n-1$ consecutive rounds.
	
\begin{algorithm}[h]
\small
\begin{algorithmic}[1]
\INITIALLY{}
  \STATE $x_p \gets$ initial position of $p$
  \STATE $m_p \gets x_p$; $M_p \gets x_p$ 
\ROUND{}
  \STATE send $(m_p,M_p)$ to all agents in $\Out_p(t)$ and receive $(m_q,M_q)$ from all agents $q$ in $\In_p(t)$
  \STATE $m_p \gets \min\big\{ m_q \mid q \in \In_p(t)\big\}$;  $M_p \gets \max\big\{ M_q \mid q \in \In_p(t)\big\}$
  \IF{$t \equiv 0 \mod n-1$}
    \STATE $x_p \gets (m_p+M_p)/2$
    \STATE $m_p \gets x_p$; $M_p \gets x_p$
  \ENDIF
\end{algorithmic}
\caption{Amortized MidPoint algorithm for agent $p$}
\label{algo:mid}
\end{algorithm}

\begin{thm}\label{thm:amormp}
In a rooted network model, the Amortized MidPoint algorithm solves 
	asymptotic consensus with  convergence rate $\varrho \leq 1- \frac{1}{2n}$ 
	and  convergence time $ T(\varepsilon) \leq ( n-1) \left\lceil \log_2  \frac{\delta(0)}{\varepsilon}  \right\rceil $.
\end{thm}

Under the assumption that all agents know~$n$,
	the Amortized MidPoint  algorithm thus solves asymptotic consensus in linear-time and 
	with only two  values per  agent and per message.
A similar result has been recently obtained by Olshevsky~\cite{Ols15} with a linear-time algorithm,
	but this algorithm works only with a fixed communication graph that further ought to be 
	bidirectional and connected.

\section{\!\!Component-Wise Algorithms for the Multi-Dimensional Case}\label{sec:multi_component}

We now tackle the problem of  multi-dimensional asymptotic consensus,
	and present several algorithms that are all generalizations of MidPoint
	to higher dimension.
For the analysis of these algorithms, we proceed component by component:
	a $d$-dimensional execution is equivalent to $d$ one-dimensional executions.
In particular, we extend the property of $\alpha$-safeness, $\alpha \in [0,1/2]$, to a higher dimension by
	enforcing~(\ref{eq:safeness}) along each dimension.
Formally, a convex combination algorithm in dimension~$d$ is {\em 
	$\alpha $-safe\/} if for any~$t\in \IN$,
	\begin{align}
 	 \alpha \, M_{p,i}(t) + (1-\alpha) \, m_{p,i}(t) \leq x_{p,i}(t+1) \leq  (1-\alpha) \, M_{p,i}(t) + \alpha \, m_{p,i}(t)\label{eq:alpha:higher}
	\end{align}
	where $m_{p,i}(t)$ is the minimum and $M_{p,i}(t)$ the maximum of the values 
	$\{x_{p,i}(t) \mid q\in\In_p(t+1) \}$ in the $i$\textsuperscript{th} component of
	the positions of the in-neighbors of~$p$ in round~$t+1$, respectively.
Although this definition syntactically depends on the chosen coordinate system, it is in fact coordinate-free.
This can be seen by applying, to the set of agent positions, the inverse
of the transformation taking one coordinate system to another.
Also note that, in contrast to the one-dimensional case, \eqref{eq:alpha:higher}
	does {\em not\/} guarantee that the algorithm is a convex combination algorithm.
	
With this definition, Theorem~\ref{thm:nonsplit} holds in higher dimension.
Its proof is exactly the same, applied in each component.

Like in one dimension, one may use the amortization technique in higher dimension to go
	from nonsplit  to rooted network models by paying a
	multiplicative price of $n-1$ in terms of convergence time.
It requires all agents to know the size~$n$ 
	of the network and applies to  convex combination algorithms in which multiplicity is not taken into
	account in the weights of the position update rules.
Also, it requires  {\em a priori\/} to increase channel bandwidth and local storage capacities by a factor~$n$.


\subsection{Asymptotic Consensus in Dimension Two}\label{sec:dimtwo}

A component-wise application of the MidPoint algorithm is obviously  $1/2$-safe.
Unfortunately the following example shows that it is not 
	a convex combination algorithm when $d\geq3$, and thus may violate the validity
	clause: the convex hull of the points $(1,0,0)\,,\,(0,1,0)\,,\, (0,0,1)$ in~$\IR^3$ 
	does not contain the 
	component-wise midpoint $M=(1/2\,,\,1/2\,,\,1/2)$.
	
Nonetheless,  the following lemma shows that taking the component-wise midpoint does
	not exit the convex hull in dimension two.

\begin{lem}\label{lem:component:midpoint:r2}
Let~$C$ be a nonempty compact convex set in~$\IR^2$.
Then $\big( \frac{x^+_1 - x^-_1}{2}\,,\,\frac{x^+_2 - x^-_2}{2} \big) \in C$ where
$$
x^+_i = \max \{x_i \mid \exists y \in C \colon y_i = x_i\}   \  \mbox{ and }\ 
x^-_i = \min \{x_i \mid \exists y \in C \colon y_i = x_i\}   \, .
$$
\end{lem}
\begin{proof}
Without loss of generality, we assume $x_i^-  = 0$ and $x_i^+  = 1$ for $i=1,2$ by scaling 
	and translation, and we shall show  that $m = (1/2,1/2)\in C$.

Let $a = (0,a_2)\in C$ be a point with minimal first component
	and $b = (b_1,0)\in C$ one with minimal second component.
Intersecting the segment $\left\{ \big(\lambda b_1 , (1-\lambda)a_2\big) \mid \lambda\in[0,1] \right\}$
	that joins~$a$ to~$b$ with the first median, we get the point $c \in C$ with coordinates:
	$$c = \left \{ \begin{array}{ll}
	\left( \frac{b_1a_2}{b_1+a_2} , \frac{b_1a_2}{b_1+a_2} \right) & \mbox{ if } b_1\neq 0 \mbox{  or } a_2\neq 0, \\ \\
	(0,0) & \mbox{  if } b_1=a_2=0 \, .
	\end{array} \right. $$
In both cases, since $b_1a_2\leq \min\{b_1,a_2\}\leq b_1+a_2$, we have  $c= (\alpha, \alpha)$ with $\alpha\leq 1/2$.

A symmetric argument for two points with maximal coordinates yields a point $c'$ in $C$ 
	such that $c'= (\beta, \beta)$  with $\beta \geq 1/2$.
Observing that 
	$$ \frac{1}{2} = \frac{\beta - 1/2}{\beta-\alpha}\cdot \alpha +  \frac{1/2 - \alpha}{\beta-\alpha} \cdot \beta \, ,$$
	we then write $m$ as a convex combination of the two points $c$ and $c'\!$, which shows that
	$m$  is in~$C$.
\end{proof}

Consequently,  the component-wise MidPoint algorithm actually {\em is\/} a convex 
	combination algorithm in dimension two.
By analyzing each component separately, our results on the MidPoint algorithm carry 
	over from the one-dimensional to the two-dimensional case.
In particular,  we can apply the amortization technique, which yields the following result.

\begin{thm}\label{thm:mid:2}
In the particular case of dimension two, the component-wise MidPoint algorithm
	solves asymptotic consensus in any rooted network model with convergence rate 
	$\varrho \leq 1 - \frac{1}{2n}$ and convergence time
	$T(\varepsilon)  \leq (n-1) \left\lceil \log_{2} \frac{\delta(0)}{\varepsilon}  \right\rceil$.
\end{thm}

Observe that the component-wise mid-point depends on the chosen coordinate system.

%
%
%
%
%

\subsection{The ExtremePoint Algorithm}\label{sec:extremepoint}

We now introduce an algorithm, called the {\em ExtremePoint algorithm},
	that generalizes the MidPoint algorithm in arbitrary dimension.
In this algorithm, every agent collects its in-neighbors' positions,
	identifies among them two extreme points in each component, 
	and then averages over these $2d$ extreme positions.
	
For each component, the update rule is an average of exactly~$2d$ real numbers.
We thus easily check that the ExtremePoint algorithm is $1/(2d)$-safe.
From Theorem~\ref{thm:nonsplit}, we derive that  in any nonsplit network model,
	the ExtremePoint algorithm achieves asymptotic consensus with  
	$\varrho \leq 1-\frac{1}{2d}$ and
	$T(\varepsilon) \leq \left \lceil \log_{2d/(2d-1)}  \frac{\delta(0)}{\varepsilon} \right \rceil$.

As with MidPoint, the weights in  the ExtremePoint algorithm 
	depend only on the {\em sets\/} of received positions without any multiplicity.
The algorithm thus admits an  amortized version given in Algorithm~\ref{algo:mid:2d}.
During the position-gathering  phase, $p$ keeps track of the positions 
	of  two in-neighbors with the smallest and
	the largest $i^{\text{th}}$ component, for every component~$i$.
Hence, $p$ records exactly~$2d$ points in each round,  a number 
	independent of $n$.
Then $p$ moves to the centroid of these $2d$ extreme points.

\begin{algorithm}[h]
\small
\begin{algorithmic}[1]
\INITIALLY{}
  \STATE $x_p \gets$ initial position of~$p$
  \STATE $m_p^{(1)},m_p^{(2)},\dots, m_p^{(d)} \gets x_p$; $M_p^{(1)},M_p^{(2)},\dots, M_p^{(d)} \gets x_p$ 
\ROUND{}
  \STATE send $\big(m_p^{(1)},\dots, m_p^{(d)},M_p^{(1)},\dots, M_p^{(d)}\big)$
	to all agents in $\Out_p(t)$ and\\receive $\big(m_p^{(1)},\dots,
m_p^{(d)},M_p^{(1)},\dots, M_p^{(d)}\big)$ from all agents $q$ in $\In_p(t)$
  \FOR{$i\gets 1$ to $d$}
    \STATE $m_p^{(i)} \gets m_q^{(i)}$ with minimal $i$\textsuperscript{th} component $m_{q,i}^{(i)}$ where $q\in\In_p(t)$
    \STATE $M_p^{(i)} \gets M_q^{(i)}$ with maximal $i$\textsuperscript{th} component $M_{q,i}^{(i)}$ where $q\in\In_p(t)$
  \ENDFOR
  \IF{$t \equiv 0 \mod n-1$}
    \STATE $x_p \gets \frac{1}{2d}\left( 
	\sum_{i=1}^d m_p^{(i)} +
	\sum_{i=1}^d M_p^{(i)} 
	\right)$
    \FOR{$i\gets 1$ to $d$}
      \STATE $m_p^{(i)} \gets x_p$; $M_p^{(i)} \gets x_p$
    \ENDFOR
  \ENDIF
\end{algorithmic}
\caption{Amortized ExtremePoint algorithm for agent $p$}
\label{algo:mid:2d}
\end{algorithm}

By Proposition~\ref{prop:productrooted}, the communication graph in each macro-round of $n-1$ 
	rounds is nonsplit.
Combined with Theorem~\ref{thm:nonsplit}, we obtain the following theorem.

\begin{thm}\label{thm:mid:2d}
In a rooted network model, the Amortized ExtremePoint algorithm
	solves asymptotic consensus with convergence rate $\varrho \leq 1 - \frac{1}{2dn}$
	and convergence time
	$T(\varepsilon)  \leq (n-1) \left\lceil \log_{2d/(2d-1)} \frac{\delta(0)}{\varepsilon} \right\rceil$.
\end{thm}

Observe that in the above algorithm, new positions at each round (line 9) depend
	both on non-deterministic choices for the points $m_p^{(i)}$ and $M_p^{(i)}$
	(lines 5--6) and on the chosen coordinate system.

\section{The Multi-Dimensional Case: A Coordinate-Free Algorithm}\label{sec:centroid}

Both asymptotic consensus algorithms presented in Section~\ref{sec:multi_component}
 	treat agent positions component-wise, thus intrinsically assuming a \emph{common, 
 	agreed-upon coordinate system}.
The same applies for the work on multidimensional approximate consensus~\cite{MH13} 
	where convergence is obtained by cycling through the coordinate components, converging component by component.
While the assumption of a common coordinate system, depending on the application,
 	may be plausible in some man-made systems, the assumption is highly questionable
 	in natural systems such as swarms of birds or bacteria and social models in opinion dynamics.
	
We now present the {\em Centroid algorithm}, a generalization of the MidPoint algorithm 
	that is \emph{coordinate-free} in the sense that it does not require an a priori agreed-upon 
	coordinate system:
Each agent moves to the centroid of the {\em convex hull\/} of the positions of its in-neighbors
	in the current communication graph, with uniform mass distribution over the convex hull.  
While the ExtremePoint algorithm computes the centroid of a {\em finite\/} set of points with equal mass, 
	the Centroid algorithm computes the centroid of the whole convex hull of these points.
	
The main point of this section is to show that by spreading the mass to the convex hull, we obtain
        an algorithm that is $1/(d+1)$-safe.
We give the proof sketch in Section~\ref{sec:safeness_proof}.
        
\begin{thm}\label{thm:centroid}
The Centroid algorithm is a $1/(d+1)$-safe convex combination algorithm.
\end{thm}


From Theorem~\ref{thm:nonsplit} we thus obtain a convergence rate of $1-\frac{1}{d+1}$ in nonsplit network models instead of $1-\frac{1}{2d}$ for the ExtremePoint algorithm.
Since the algorithm's update rule does not take into account any multiplicity, the Centroid algorithm
        admits an  amortized version  given in Algorithm~\ref{algo:mid:higher}.
We use $\hull(A)$ to denote the convex hull of a set $A \subseteq \IR^d$.
 
\begin{algorithm}[h!]
\small
\begin{algorithmic}[1]
\INITIALLY{}
  \STATE $x_p \gets$ initial position of $p$
  \STATE $C_p \gets \{x_p\}$ 
\ROUND{}
  \STATE send $C_p$ to all agents in $\Out_p(t)$ and receive $C_q$ from all agents $q$ in $\In_p(t)$
  \STATE $C_p \gets C_p \cup \bigcup_{q\in \In_p(t)} C_q$
  \IF{$t \equiv 0 \mod n-1$}
    \STATE $x_p \gets $ centroid of $\hull(C_p)$
    \STATE $C_p \gets \{x_p\}$
  \ENDIF
\end{algorithmic}
\caption{Amortized Centroid algorithm for agent $p$}
\label{algo:mid:higher}
\end{algorithm}

From Proposition~\ref{prop:productrooted} and Theorem~\ref{thm:centroid} we finally obtain the following result.
\begin{thm}\label{thm:algo:4}
In a rooted network model, the Amortized Centroid algorithm solves asymptotic consensus
	with convergence rate $\varrho \leq 1-\frac{1}{n\cdot (d+1)}$
 	and convergence time $T(\varepsilon) \le (n-1)\left\lceil\log_{(d+1)/d} \frac{\delta(0)}{\varepsilon} \right\rceil$.
\end{thm}


While the centroid of a body $A$ cannot be efficiently computed in general, 
	we are in the case of $A$ being a convex bounded $d$-polytope with at most 
	$n$ vertices.
Although exact computation of the centroid has been shown to be $\#P$-hard even for 
	these bodies~\cite{Rad07}, polynomial (in $n$) algorithms based on simplex decompositions
	exist if one fixes the dimension~$d$.
Besides, natural systems may be equipped with natural means to determine centroids.
Since the Amortized Centroid algorithm relays all positions during its gathering phase,
        it a priori requires capabilities to store and relay up to $n$ positions per round.
This is in contrast to the MidPoint and the ExtremePoint Amortized algorithms.
Optimizations, however, exist that may pay off in certain applications:
        in code line~4, the non-extreme points of $\hull(C_p)$ can be removed from $C_p$.
While the frame, i.e., the set of extreme points, can be computed in polynomial time
        by solving linear programs~\cite{DH96},
        one may not be willing to pay this additional overhead in each round.
Alternatively, computationally less intensive heuristics can be applied
        to remove many of the non-extreme points, see, e.g., \cite{DT81}.

\subsection{Safeness Proof}\label{sec:safeness_proof}
We now tackle the proof of Theorem~\ref{thm:centroid}.
First let us introduce some notation.
Let us denote the $d$-dimensional volume of set $A \subseteq \IR^d$ by $\vol_d(A)$.
For $A \subseteq \IR^d$ and $j \in [d]$, let $m_j(A) = \inf_{x \in A}x_j$ and $M_j(A) = \sup_{x \in A}x_j$.
We next define sets, representing geometric bodies, that are symmetric around the first axis.
For each $\xi \in \IR$, let $H_\xi = \{x \in \IR^d \mid x_1 = \xi\}$ be the hyperplane
  in $\IR^d$ orthogonal to the first axis, intersecting it at $(\xi,0,\dots,0)$.
Let $\Cube_\xi(\gamma)$ be the $(d-1)$-cube of edge length~$\gamma$ that lies within
  hyperplane~$H_\xi$ and is centered at point $(\xi, 0, \dots, 0)$, i.e.,
$\Cube_\xi(\gamma) = \{ x\in \IR^d \mid x_1 = \xi \wedge \max_{2\le j\le d}|x_j| \le \gamma/2\}$.
For a function $\ell: \IR \to \IR_0^+$ we define the {\em symmetric body $S(\ell)$\/} as
\begin{align}\label{eq:Symm}
  S(\ell) = \bigcup_{\xi \in \IR}\Cube_\xi(\ell(\xi))\,.
\end{align}

Roughly speaking, we proceed as follows.
Each $\hull(C_p)$ in code line~6 is a bounded convex polytope in $\IR^d$.
Fix agent $p$ and component $i$ along which $\alpha$ in \eqref{eq:alpha:higher} is minimized.
We then use a Steiner-type symmetrization along the $i^{\text{th}}$ axis:
We transform polytope~$A = \hull(C_p)$ into polytope~$A' = \hull(C'_p)$ which is highly symmetric around the $i^{\text{th}}$ axis
  and whose $i^{\text{th}}$ centroid component is invariant under the transformation.
Figure~\ref{fig:symm} depicts the idea of the transformation in dimension two: the symmetric body $A'$ is constructed
  such that cuts orthogonal to the first axis of $A'$ have same volume as their corresponding cuts in $A$.
This ensures invariance of the first component of the centroid~$c$.
We then reduce the problem to the class of those $A'$ that are formed by a hyperpyramid extended by a $d$-box at its base.
Among these we show the hyperpyramids without $d$-boxes to minimize $\alpha$ in \eqref{eq:alpha:higher},
  finally reducing $A'$ to hyperpyramids.
From a lower bound on the distance of $\centroid_i(A')$ to its base, and the fact that the involved transformations and reductions
  did not shift $\centroid_i(A')$ away from its base, we are finally able to prove safeness.
  
\begin{figure}
\centering
\includegraphics{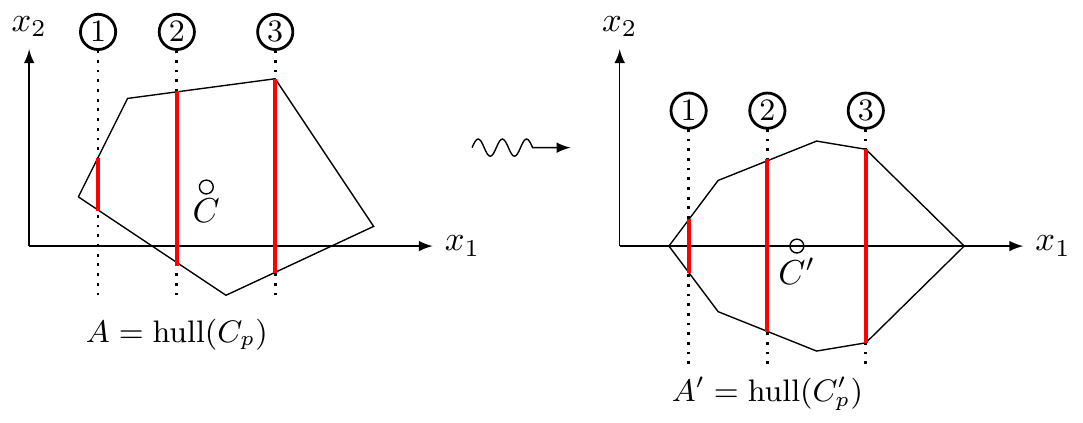}
\caption{Symmetrization of the polytope $A = \hull(C_p)$ around the first axis. The original polytope $A$ is transformed into the
  symmetric polytope $A' = \hull(C'_p)$ such that cuts orthogonal to the first axis have same volume.
  The transformation ensures invariance of the first component of the centroid~$C$.}\label{fig:symm}
\end{figure}  

We start with some auxiliary lemmas on symmetrized bodies.
First we show that if one removes parts from a body whose first component are left of the body's
  centroid, then the first component of the centroid moves to the right.
\begin{lem}\label{lem:shift}
Let $\ell,\ell': \IR \to \IR_0^+$.
If $\ell'(\xi) \le \ell(\xi)$ for $\xi \le \centroid_1(S(\ell))$ and
  $\ell'(\xi) = \ell(\xi)$ for $\xi > \centroid_1(S(\ell))$
  then $\centroid_1(S(\ell')) \ge \centroid_1(S(\ell))$.
\end{lem}
\begin{proof}
Abbreviate $c = \centroid_1(S(\ell))$ and $c' = \centroid_1(S(\ell'))$.
It is
\begin{align*}
  c = \frac{\int_{-\infty}^{c}\xi\ell(\xi)^{d-1}d\xi + \int_{c}^{\infty}\xi\ell(\xi)^{d-1}d\xi}{
            \int_{-\infty}^{c}\ell(\xi)^{d-1}d\xi + \int_{c}^{\infty}\ell(\xi)^{d-1}d\xi}\,.
\end{align*}
Algebraic manipulation yields
\begin{align}
\int_{-\infty}^{c}(c-\xi)\ell(\xi)^{d-1}d\xi =
  \int_{c}^{\infty}(\xi-c)\ell(\xi)^{d-1}d\xi \,.\label{eq:shift1}
\end{align}
Because $c-\xi \ge 0$ whenever $\xi \leq c$ and since $\ell'(\xi) \le \ell(\xi)$ for those~$\xi$, we get
\begin{align*} 
  \int_{-\infty}^{c}(c-\xi)\ell'(\xi)^{d-1}d\xi \le
\int_{-\infty}^{c}(c-\xi)\ell(\xi)^{d-1}d\xi\,.
\end{align*}
Together with \eqref{eq:shift1} this gives
\begin{align*} 
  \int_{-\infty}^{c}(c-\xi)\ell'(\xi)^{d-1}d\xi \le
\int_{c}^{\infty}(\xi-c)\ell(\xi)^{d-1}d\xi\,.
\end{align*}
Again, algebraic manipulation yields
\begin{align*}
  c \le \frac{\int_{-\infty}^{c}\xi\ell'(\xi)^{d-1}d\xi + \int_{c}^{\infty}\xi\ell(\xi)^{d-1}d\xi}{
              \int_{-\infty}^{c}\ell'(\xi)^{d-1}d\xi + \int_{c}^{\infty}\ell(\xi)^{d-1}d\xi}\,.
\end{align*}
The latter term is, in fact, equal to $c'$ because $\ell(\xi) = \ell'(\xi)$ for
$\xi > c$. 
This shows $c\leq c'$.
\end{proof}  

Lemma~\ref{lem:shift} will play a crucial role when proving that among all symmetric bodies,
  we can restrict our attention to those symmetric bodies composed of a hyperpyramid and a $d$-box.
 
The following lemma finds the one body in this class that moves the centroid the furthest to the right,
  i.e., away from the apex.
\begin{lem}\label{lem:max_centr}
Let $L \ge 0$.
For $h \in [0,L]$, and $\vartheta > 0$ let  
$\ell_{h,\vartheta}: \IR \to \IR_0^+$ be the function with
\begin{align}
  \ell(\xi) =
  \begin{cases}
    0 & \text{ if } \xi \le 0 \text{ or } \xi > L\\
    \frac{\xi\vartheta}{h} & \text{ if } \xi \in (0,h]\\
    \vartheta & \text{ if } \xi \in (h,L]\,.
  \end{cases}\notag
\end{align}
Among all symmetric bodies $S(\ell_{h,\vartheta})$ with $h \in [0,L]$, and $\vartheta > 0$,
  the symmetric bodies $S(\ell_{L,\vartheta})$, with arbitrary $\vartheta > 0$, maximize the first centroid component.
It is $\centroid_1(S(\ell_{L,\vartheta})) = \frac{Ld}{d+1} = \frac{M_1(S(\ell_{L,\vartheta}))d}{d+1}$.
\end{lem}
\begin{proof}
We first observe that $S(\ell_{h,\vartheta})$ can be decomposed into a hyperpyramid, without a base, of height~$h$ and
  a $d$-box as follows:
Define $P(h)$ by 
\begin{align*}
  P(h) = \bigcup_{\xi \in [0,h)}\Cube_\xi(\ell_{h,\vartheta}(\xi))\,.
\end{align*}
We observe that $P(h)$ is the hyperpyramid, without base, in $\IR^d$
  whose height is~$h$,
  whose apex is at $(0,\dots,0)$ and whose base is the
  $(d-1)$-cube with side length~$\vartheta$, centered at the first axis and lying within the
  hyperplane~$H_h$.
Define $B(h')$, with $h' = L-h$, by
\begin{align*}
  B(h') = \bigcup_{\xi \in [h,L]}\Cube_\xi(\ell_{h,\vartheta}(\xi)) = \bigcup_{\xi \in (h,L]}\Cube_\xi(\vartheta)\,.
\end{align*}
We observe that $B(h')$ is a $d$-box in $\IR^d$, with
\begin{align*}
  \vol_d(B(h')) = \vartheta^{d-1}(L-h)\,.
\end{align*}
It is $P(h) \cap B(h') = \emptyset$ and $S(\ell_{h,\vartheta}) = P(h) \cup B(h')$.

\medskip

We will next compute the first component of the centroids of $P(h)$ and $B(h')$, allowing us to
   compute the first component of the centroid of $S(\ell_{h,\vartheta})$.
For any $\alpha \in (0,h]$, the cut $X_\alpha = P(h) \cap H_\alpha = \Cube_\alpha(\ell_{h,\vartheta}(\alpha))$ 
   has volume $\vol_{d-1}(X_\alpha) = \left(\frac{\alpha\vartheta}{h}\right)^{d-1}$.
From the volume of a pyramid in $\IR^d$, we obtain $\vol_d(P(h)) = \frac{h\vartheta^{d-1}}{d}$.
The first centroid component of $P(h)$ thus is at
\begin{align*}
x'_1 = \frac{1}{\vol_d(P(h))}\int_{0}^{h}\alpha \vol_{d-1}(X_\alpha) \,d\alpha = \frac{hd}{d+1}\enspace.\notag
\end{align*}
By symmetry arguments the first centroid component of $B(h')$ is at
\begin{align*}
x''_1 = h+\frac{h'}{2}\enspace.\notag
\end{align*}
The first centroid component of the combined geometric body $P(h) \cup B(h')$ thus is at
\begin{align}
  x_1 = \frac{x'_1\vol_d(P(h)) + x''_1\vol_d(B(h'))}{\vol_d(P(h) \cup B(h'))}
    = \frac{d(L^2d+L^2-dh^2+h^2)}{2(d+1)(Ld-dh+h)}\,.\label{eq:L}
\end{align}

\medskip

We next distinguish between two cases for dimension~$d$:
\begin{enumerate}
\item  For $d = 1$, we obtain from \eqref{eq:L} that $x_1(h) = L/2$; and the lemma follows.

\item Otherwise, $d \ge 2$.
Algebraic manipulation yields,
\begin{align*}
\frac{dx_1(h)}{dh} = \frac{d}{2(d+1)}\left(1-\frac{L^2}{(Ld-dh+h)^2}\right) > 0\,,\notag
\end{align*}
for $0 \le h < L$ and $d \ge 2$.
Thus, $\max_{h \in [0,L]}x_1(h) = x_1(L) = \frac{Ld}{d+1}$;
  and the lemma follows also in this case.
\end{enumerate}
\end{proof}

We are now in position to show our major result on the Centroid algorithm in Theorem~\ref{thm:centroid}:
        we prove that for any convex bounded polytope $A$ in $\IR^d$ and for every $j \in [d]$, we
	have 
\begin{align}\label{eq:centroid}
  \left(1-\frac{d}{d+1}\right)M_j(A) + \frac{d}{d+1}m_j(A) \le \centroid_j(A)
\le \left(1-\frac{d}{d+1}\right)m_j(A) + \frac{d}{d+1}M_j(A)\,.
\end{align}

Choose an arbitrary component $j \in [d]$.
Without loss of generality assume that $j=1$, $m_j(A)=0$ and $M_j(A) > 0$.
It suffices to prove the right inequality in \eqref{eq:centroid} to show \eqref{eq:centroid}
   by the following argument: assume by means of contradiction that
   that the right inequality is valid for all $A$, but there is an $A$ for which the left is invalid.
Then negating all first components of points in $A$ yields a polytope that violates
   the right inequality; a contradiction to the initial assumption.
It thus suffices to show
   \begin{align}
   \centroid_1(A) \le \frac{d}{d+1}M_1(A)\,.\label{eq:toshow}
   \end{align}
   
\medskip

We now construct symmetrized body $A_s$ from $A$
  that has the same volume and the same first centroid component
  as $A$.
For that purpose we do a Steiner-type symmetrization of~$A$.

By a simple reduction to a smaller dimension, we may assume $\vol_d(A) > 0$.
Let $v_\xi = \vol_{d-1}(H_\xi \cap A)$, and let $\ell_A$ be the function $\IR \to \IR_0^+$
with
\begin{align}
  \ell_A(\xi) = \begin{cases}
    v_\xi^{1/(d-1)} & \text{ if } v_\xi > 0\,,\\
    0             & \text{ if } v_\xi = 0\,.
  \end{cases}\notag
\end{align}
Then let $A_s = S(\ell_A)$.
From \eqref{eq:Symm}, we have
  $\vol_{d-1}(H_\xi \cap A_s) = \vol_{d-1}(\Cube_\xi(\ell_A(\xi))) = \ell_A(\xi)^{d-1} = v_\xi$. 
Thus
  $\vol_{d-1}(H_\xi \cap A_s) = \vol_{d-1}(H_\xi \cap A)$ 
and further,
  $\vol_{d}(A) = \int_{-\infty}^{+\infty}\vol_{d-1}(A \cap H_\xi)\ d\xi
              = \int_{-\infty}^{+\infty}\vol_{d-1}(A_s \cap H_\xi)\ d\xi =  \vol_{d}(A_s)$.
Combining both yields,
\begin{align*}
  \centroid_1(A) = \frac{1}{\vol_d(A)}\int_{-\infty}^{+\infty}\xi\vol_{d-1}(A \cap H_\xi)\ d\xi = \centroid_1(A_s)\,,
\end{align*}
i.e., the first component of the centroid is invariant under symmetrization.
For ease of notation, abbreviate the first component of the centroid by $c = \centroid_1(A) \ge 0$.

Figure~\ref{fig:symm2} depicts the process of symmetrization around the first axis at an example
  in dimension two: the body $A'$ on the right is constructed from $A$ by symmetric 1-cubes (line segments)
  centered at $\xi \in [m_1(A),M_1(A)]$ such that their volume (length)
  $v_\xi = \vol_1(H_\xi \cap A) = \vol_1(H_\xi \cap A')$.
  
\begin{figure}[ht]
\centering
\includegraphics{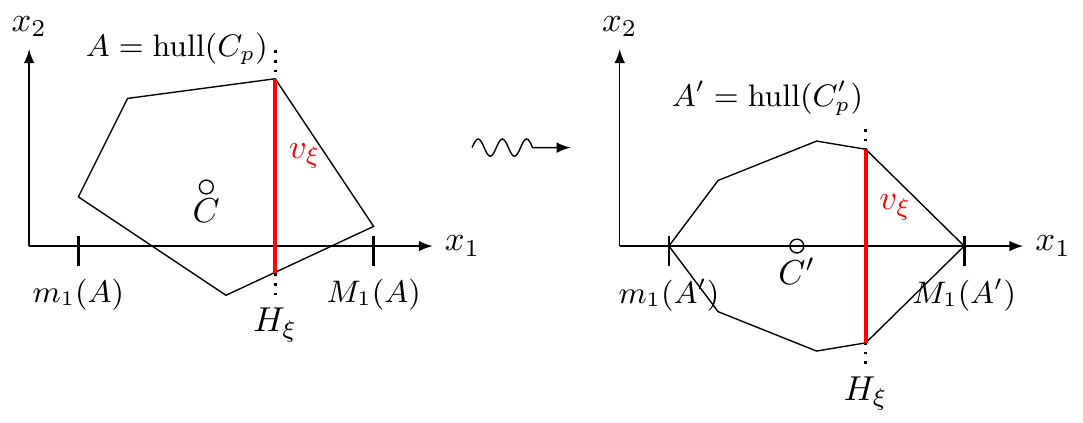}
\caption{Symmetrization of the polytope $A = \hull(C_p)$ around the first axis. The original polytope $A$ is transformed into the
  polytope $A' = \hull(C'_p)$ such that $m_1$, $M_1$ and the first component
  of the centroid~$C$ are invariant under the transformation.}\label{fig:symm2}
\end{figure}  

First observe, that $H_\xi \cap A = \emptyset$ for all $\xi < m_1(A) = 0$ and all $\xi > M_1(A)$.
From the fact that $A$ is a bounded polytope, $M_1(A) < \infty$ and $\ell_A$ is bounded.
Thus $\ell_A$ is zero outside of~$[0,M_1(A)]$.

We next show that function~$\ell_A$ is concave in $[0,M_1(A)]$.
In fact, concavity of $\ell_A$ in $[m_1(A),M_1(A)]$ is equivalent to convexity of~$A_s$.

Arbitrarily choose~$\alpha,\beta \in [0,M_1(A)]$, and $t\in [0,1]$.
From convexity of~$A$ we have,
  $(H_{t\alpha+(1-t)\beta} \cap A) \supseteq t(H_{\alpha} \cap A) +  (1-t)(H_{\beta} \cap A)$ 
where  ``$+$'' on the right side denotes the Minkowski sum of sets.
Hence
\begin{align}
  \ell_A(t\alpha+(1-t)\beta)^{d-1} = \vol_{d-1}(H_{t\alpha+(1-t)\beta} \cap A)
    \geq \vol_{d-1}\big(t(H_{\alpha} \cap A) + (1-t)(H_{\beta} \cap A)\big)\,.\label{eq:conv}
\end{align}
Applying the Brunn-Minkowski inequality we obtain,
\begin{align}
\vol_{d-1}\big(t(H_{\alpha} \cap A) + (1-t)(H_{\beta} \cap A)\big)^{1/(d-1)} \ge t\vol_{d-1}\big(H_{\alpha} \cap A\big)^{1/(d-1)} + (1-t)\vol_{d-1}\big(H_{\beta} \cap A\big)^{1/(d-1)}\,.\notag
\end{align}
Together with~\eqref{eq:conv} this yields,
  $\ell_A(t\alpha+(1-t)\beta) \ge t\ell_A(\alpha) + (1-t)\ell_A(\beta)$, 
i.e., the concavity of function~$\ell_A$ in $[0,M_1(A)]$.

\medskip

By the following reduction argument, we can further assume that~$\ell_A$ has constant positive slope
  within~$[0,c]$, i.e., for all $\xi \in [0,c]$,
  $\ell_A(\xi) = \xi c/\ell_A(c)$.   
Assume by means of contradiction that this is not the case and
  consider the continuous function $\ell'_A: \IR \to \IR_0^+$ with
\begin{align}
  \ell'_A(\xi) = \begin{cases}
        \xi c/\ell_A(c) & \text{ if } \xi \in [0,c]\,,\\
        \ell_A(\xi) & \text{ else }\,.
        \end{cases}
\end{align}
Note that $\ell_A$ and $\ell'_A$ differ only within $[0,c]$, and that $\ell'_A$ has constant positive
  slope $c/\ell_A(c)$ within $[0,c]$.
By concavity of $\ell_A$, we have $\ell'_A(\xi) \le \ell_A(\xi)$ for all~$\xi \in [0,c]$.
From Lemma~\ref{lem:shift}, $\centroid_1(S(\ell'_A)) \ge \centroid_1(S(\ell_A))$.
We may thus consider $\ell'_A$ instead of $\ell_A$; the reduction follows.

\medskip

Let the maximum value achieved by $\ell_A$ be
  $\vartheta = \max_{\xi \in [0,M_1(A)]}\{\ell_A(\xi)\}$ 
and let $h \in [0,M_1(A)]$ be the smallest value where $\ell_A(h) =  \vartheta$, i.e., the maximum is reached.

By a reduction argument, we now show that~$\ell_A$ can be assumed to have constant positive slope
  within $[0,\max(c,h)]$ and has value~$\vartheta$ within $[\max(c,h),M_1(A)]$.
Assume by means of contradiction that this is not the case and
  consider the continuous function $\ell'_A: \IR \to \IR_0^+$ with
\begin{align}
  \ell'_A(\xi) = \begin{cases}
        \xi c/\ell_A(c) & \text{ if } \xi \in [0,\max(c,h)]\,,\\
        \vartheta & \text{ else }\,.
        \end{cases}
\end{align}
We distinguish between two cases for~$h$:

\begin{enumerate}
\item In case $h < c$, function $\ell_A$ and $\ell'_A$ may differ only within $[c,M_1(A)]$.
By definition of $\vartheta$, it holds that $\ell_A(\xi) \leq \vartheta = \ell'_A(\xi)$ for all $\xi \in [c,M_1(A)]$.
We may thus apply Lemma~\ref{lem:shift}, and obtain that
  $\centroid_1(S(\ell'_A)) \ge \centroid_1(S(\ell_A))$; the reduction follows in this case.

\item Otherwise, $h\ge c$ and function $\ell_A$ and $\ell'_A$ may differ only within $[h,M_1(A)]$.
By concavity of $\ell_A$, we have that $\ell'_A(\xi) \le \ell_A(\xi)$ for all~$\xi \in [c,h]$.
Further, by definition of $\vartheta$, we have $\ell_A(\xi) \leq \vartheta = \ell'_A(\xi)$
  for all $\xi \in [h,M_1(A)]$.
Again, we apply Lemma~\ref{lem:shift} and obtain that
  $\centroid_1(S(\ell'_A)) \ge \centroid_1(S(\ell_A))$; the reduction also follows in this case.

\end{enumerate}
  
We thus obtain that $\ell_A$ is of the form as required by Lemma~\ref{lem:max_centr},
  with $L = M_1(A)$, $h$ and $\vartheta$.
This yields \eqref{eq:toshow}, which concludes the proof.

\section{ExtremePoint and Centroid with Disconnectivity}\label{sec:moreau}

The aim of this section is to study the behavior of the ExtremePoint and the Centroid 
	algorithms under very weak connectivity assumptions.
Namely, we prove that the striking convergence properties of the convex combination 
	algorithms with {\em non-vanishing\/} and {\em bounded\/} 
	weights (e.g., EqualNeighbor) extend to both the ExtremePoint and the Centroid 
	algorithms and, more generally, to  every convex combination algorithm that is 
	$\alpha$-safe. 
The result is based on the fundamental convergence theorem on infinite product of stochastic matrices proved 
	by Moreau in~\cite{Mor05} that we recall now.
	
Let $\big(A(t)\big)_{t\in \IN^*}$ be a sequence of stochastic matrices of size $n$ and let
	$G(t)$ denote  the directed graph associated to $A(t)$.
The edges that appear infinitely often in the directed graphs $ G(t)$ define a directed graph
 	denoted~$G^{\infty}$.
The	following assumptions are made about the matrices~$A(t)$:
	\begin{description}
	\item[A1] Each matrix $A(t)$ has a positive diagonal, i.e.,
	$A_{p p}(t) >0$ for all $p \in [n]$.
	\vspace{-0.2cm}
	\item[A2] There exists some $a \in ]0,1]$ such that $A_{p q}(t) \in \{0\}\cup [a,1]$
		for all $p,q \in [n]$ and all $t\in \IN^*$.
	\vspace{-0.2cm}
	\item[A3] For each $t\in \IN^*$, the directed graph $G (t) $ is bidirectional.
	\vspace{-0.2cm}
	\item[A4] The directed graph $ G^{\infty}$ is strongly 
	connected.
	\end{description}
	
\begin{thm}[{\cite{Mor05}}]\label{thm:moreau}
Under assumptions A1--A4, the left-infinite product of stochastic matrices $\prod_{t=1}^{\infty} A(t) $ converges
	to a stochastic matrix with identical rows.
\end{thm}

This theorem is remarkable because it shows that in the case of bidirectional interactions, {\em without 
	any connectivity assumptions}, every convex combination algorithm with non-vanishing and bounded 
	weights converges and achieves asymptotic consensus among agents that are not disconnected from some time~on.
	
Indeed, let $\big(G_t \big)_{t\geq 1}$  be a communication pattern composed of bidirectional directed graphs
	such that  the directed graph of the edges that appear infinitely often is strongly connected.
In other words, the agents are infinitely often connected.
We consider a convex combination algorithm with non-vanishing weights that are lower bounded by some $a >0$,
	i.e., 
	\begin{equation}\label{eq:weights}
	\forall (p,q) \in E_t, \ w_{p q}(t) \geq a \, .
	\end{equation}

Let $W(t)$ denote the $n \times n$ stochastic matrix with entries $w_{p q}(t)$.
The important point of~(\ref{eq:weights}) lies in the fact that the associated graph of the matrix~$W(t)$ 
	then coincides with the communication graph at round~$t$.
Hence assumptions A1--A4 are fulfilled by all the matrices~$W(t)$.
Theorem~\ref{thm:moreau} shows that with the  communication pattern~$\big(G_t \big)_{t\geq 1}$
	and any initial configuration $x(0) \in\big( \IR^d \big)^n$,  the convex 
	combination algorithm achieves asymptotic consensus.
	
The key point now is that in the case of dimension one, every  $\alpha$-safe averaging algorithm
	satisfies~(\ref{eq:weights}) with~$a = \alpha/n$.
	
\begin{prop}\label{prop:simplex}

Let $(v_1, \ldots, v_{n})$ any $n$-tuple of real numbers such that 
	$v_1\leq \ldots \leq v_{n}$, and let $\alpha$ be a real number in $[0, 1/2]$.
For every $x$ in the interval $[(1- \alpha)\  v_1 + \alpha \, v_n,  \alpha \, v_1 + (1- \alpha) \,  v_n]$,
	there exist $n$  real numbers $a_1, \ldots,a_{n}$ in the interval $[ \alpha / n, 1]$ such that 
	$x = a_1 \, v_1 + \ldots + a_{n} \, v_{n}$ and $a_1 + \ldots + a_{n} = 1$.
\end{prop}
\begin{proof}
Let us consider the simplex 
	$$S_n^{\alpha} = \left \{ (a_1, \dots, a_n) \in [\alpha/n, 1] \mid a_1 + \ldots a_n =1 \right \} \, , $$
	and let us denote
	$$ S_n^{\alpha} \cdot v = \left \{ a_1 v_1 + \dots + a_n v_n \mid  (a_1, \ldots, a_n) \in S_n^{\alpha}  \right \} \, .$$
We easily check that
	$$ S_n^{\alpha} = \alpha\cdot  (1/n, \dots, 1/n) + (1-\alpha)\cdot S_n^0 \, .$$
Hence $ S_n^{\alpha} \cdot v $ is the compact interval 
	$$ S_n^{\alpha} \cdot v = \left [ \alpha\,  \overline{v}  +(1-\alpha)\,  v_1\, , \,  \alpha \, \overline{v} + (1-\alpha) \, v_n \right ] $$
	where  $ \overline{v}  = (v_1+ \cdots + v_n)/n$.
Since $v_1 \leq \overline{v} \leq v_n$, we have
	$$ \alpha \, \overline{v} +(1-\alpha) \, v_1 \leq (1- \alpha) \, v_1 + \alpha \, v_n
	     \leq   \alpha \, v_1 + (1- \alpha) \, v_n \leq  \alpha\,  \overline{v} + (1-\alpha) \, v_n \, , $$
	 and so  $ S_n^{\alpha} .v $ contains the interval 
	 $ \left[ (1- \alpha)\, v_1 + \alpha\, v_n,  \alpha \,  v_1 + (1- \alpha) \,  v_n \right ] $,	which completes the proof.
\end{proof}
	
 Then we derive  the following theorem applying to the ExtremePoint and the Centroid 
	algorithms.

\begin{thm}\label{thm:bidir}
Asymptotic consensus is achieved in any execution of a convex combination algorithm that is 
	$\alpha$-safe if communication graphs are all bidirectional and if the agents are infinitely often
	connected. 
\end{thm} 
\begin{proof}
We  apply Proposition~\ref{prop:simplex} for each component:
	every round of an  $\alpha$-safe algorithm thus corresponds to a $n d \times n d$ 
	block diagonal matrix.
The $k^{\text{th}}$ block for the $k^{\text{th}}$ component is an $n  \times  n$ stochastic matrix\footnote{%
	Since Proposition~\ref{prop:simplex} is applied component-by-component, 
	there is no guarantee that for a given round~$t$, the matrices~$A_1(t), \dots, A_d(t)$ are equal.},
	denoted~$A_k(t)$, and its associated graph is exactly $ G_t^{\, -1}$.

Each matrix~$A_k(t)$ satisfies assumptions A1--A4 with $a = \alpha/n$.
By Theorem~\ref{thm:moreau}, all the agents converge to the same  position $x^* \in \IR^d$,
	and thus the convergence and agreement conditions are satisfied.
	
For the validity condition, we just observe that at round~$t$, each agent moves within the convex hull of 
	its neighbors.
Hence the  limit position $x^*$ is in the convex hull of the initial positions.
\end{proof}
  
\section{Conclusion}\label{sec:conclusion}

In this article we introduced three algorithms for multidimensional asymptotic consensus.
All three of them work in dynamic networks with directed communication graphs that may change over time,
        with fast convergence rates.
The algorithms are generalizations of the optimal MidPoint algorithm in dimension one.
Two of them, the ExtremePoint and Centroid algorithms, work in an arbitrarily high
        dimension~$d$ and are $\frac{1}{2d}$- and
        $\frac{1}{d+1}$-safe, respectively.
Our amortization technique thus makes their convergence time linear in the number of agents,
        which is optimal.

Moreover, we showed that all three algorithms solve asymptotic consensus under very weak
        connectivity assumptions in bidirectional communication graphs.

\bibliographystyle{plain}
\bibliography{agents}

\end{document}